\documentclass[12pt]{article}
\usepackage{amssymb}
\usepackage{amsfonts}
\usepackage{amsmath}
\usepackage{amsthm}
\usepackage[nohead]{geometry}
\usepackage[singlespacing]{setspace}
\usepackage[bottom]{footmisc}
\usepackage{indentfirst}
\usepackage{dcolumn}
\usepackage{endnotes}
\usepackage{graphicx}
\usepackage{fancyref}
\usepackage{subfig}
\usepackage{algorithm}
\usepackage[noend]{algpseudocode}
\usepackage{bbm}
\usepackage{caption}
\graphicspath{{figures/}}
\usepackage{float}
\usepackage{booktabs,caption,fixltx2e}
\usepackage[flushleft]{threeparttable}
\usepackage{rotating}
\usepackage[round]{natbib}
\usepackage{tikz}
\usepackage{appendix}
\usepackage[colorinlistoftodos]{todonotes}
\usepackage{pdflscape}
\usepackage{booktabs}
\usepackage{siunitx}
\usepackage{caption}
\usepackage{epstopdf}
\usepackage{multirow}
\usepackage{color,soul}
\usepackage{booktabs,caption}
\usepackage[normalem]{ulem}
\usepackage{enumitem}
\usepackage{float}
\usepackage[version=3]{mhchem}
\usepackage{pgf}
\usepackage{tikz}                   
\usepackage{pgfplots}
\usetikzlibrary{arrows,calc}
\usepgfplotslibrary{dateplot}
\usepackage{verbatim}
\usepackage{breqn}
\pgfplotsset{compat=1.12}
\usetikzlibrary{fillbetween}
\usepackage{filecontents}
\usepackage{multirow}
\usepackage{apalike}
\usepackage{amsfonts}
\usepackage{threeparttable}
\usepackage[normalem]{ulem}

\newcommand{\mathdash}{\relbar\mkern-9mu\relbar}

\graphicspath{{figures/}}

\usepackage{mathtools}

\makeatletter
\def\BState{\State\hskip-\ALG@thistlm}
\makeatother
\pgfplotsset{
	/pgfplots/area cycle list/.style={/pgfplots/cycle list={%
			{black,fill=yellow!20!white,mark=none},%
			{black,fill=yellow!40!white,mark=none},%
			{black,fill=red!20!white,mark=none},%
			{black,fill=red!40!white,mark=none},
			{black,fill=red, mark = none}
		}
	},
}

\usepackage{flexisym}
\makeatletter
\def\@biblabel#1{\hspace*{-\labelsep}}
\makeatother
\usepackage{authblk}

\DeclareMathAlphabet      {\mathbfit}{OML}{cmm}{b}{it}

\makeatletter
\def\BState{\State\hskip-\ALG@thistlm}
\makeatother

\usepackage[normalem]{ulem}
\newtheorem{theorem}{Theorem}
\newtheorem{remark}{Remark}
\newtheorem{corollary}{Corollary}

\newtheorem{prob}{Problem}
\newtheorem{lemma}{Lemma}
\newtheorem{exmp}{Example}[section]
\newtheoremstyle{case}{}{}{}{}{}{:}{ }{}
\theoremstyle{case}
\newtheorem{case}{Case}

\makeatletter
\renewcommand*\env@matrix[1][\arraystretch]{%
  \edef\arraystretch{#1}%
  \hskip -\arraycolsep
  \let\@ifnextchar\new@ifnextchar
  \array{*\c@MaxMatrixCols c}}
\makeatother

\DeclareMathOperator*{\argmin}{arg\,min}
\title{Time Series Graphical Lasso and Sparse VAR Estimation}
\author{Aramayis Dallakyan \qquad Rakheon Kim \qquad Mohsen Pourahmadi \qquad\\ Texas A\&M University}
\date{}

\begin{document}
	
\maketitle
\sloppy
\singlespacing

\begin{center}
\textbf{Abstract}
\end{center}
\noindent

\noindent

We improve upon the two-stage sparse vector autoregression (sVAR) method in \citet{Davis2016}  by  proposing an alternative  two-stage modified sVAR method which relies on time series graphical lasso to estimate sparse inverse spectral density in the first stage, and the second stage refines non-zero entries of the AR coefficient matrices using a false discovery rate (FDR) procedure. 
Our method has the advantage of avoiding the inversion of the spectral density matrix but has to deal with optimization over Hermitian matrices with complex-valued entries. It significantly improves the computational time with a little loss in forecasting performance.
We study the properties of our proposed method and compare the performance of the two methods using simulated and a real macro-economic dataset. Our simulation results show  that the proposed modification or msVAR is a preferred choice when the goal is to learn
the structure of the AR coefficient matrices while sVAR outperforms msVAR when
the ultimate task is forecasting.


\vspace{0.5 cm}
\strut


\vspace{0.25 cm}
\strut
\vspace{0.5 cm}
\thispagestyle{empty}

\pagebreak
\doublespacing
\setlength{\parskip}{.85mm plus.25mm minus.25mm}

\section{Introduction}

A vector autoregressive (VAR) model is a powerful tool for analyzing multivariate time series. The recent increase in the availability of time series data pivots the interest of researchers toward high-dimensional VAR models. A common strategy in high-dimensional VAR estimation is to impose regularization on AR coefficient matrices. These methods can be grouped into three different approaches: regularized least square estimators using Lasso-type penalties \citep{Song2011, basu2015, kock2015, Nicholson2016, matteo2019, safikhani2020}; regularized maximum-likelihood estimators \citep{basu2015, Davis2016, yuen2018}, and regularized Yule-Walker estimators using the CLIME or Dantzig estimators \citep{han15a, wu2016, ding2017}.

Regularized least square VAR methods ignore the contemporaneous dependence in the time series since the loss function does not
include the covariance of error terms. \citet{Song2011} discuss the possible impact in fitting a VAR model in which the contemporaneous dependence is ignored.  \citet{Davis2016} numerically show that the forecasting performance of the VAR model improves when
the information on the error covariance matrix is incorporated in the regularized log-likelihood.
They proposed a two-stage approach to fit sparse VAR models. In the first stage, instead of working in a time domain, authors resort to a frequency domain and estimate the partial spectral coherence (PSC) to identify possible non-zero autoregressive coefficients (see Section~\ref{s:psc} for details). Then, using constrained maximum likelihood estimation, parameters are estimated under the sparsity constraint. The lag order $p$ and the number of pairs of non-zero AR parameters $M$ are chosen using the Bayesian Information Criterion (BIC) over the specified grid values of $M$ and $p$. In the second stage, the selected model is refined by identifying spurious non-zero AR coefficients. In particular, for non-zero AR coefficients, a sequence of t-statistics is created, and $m$ of them are chosen using the BIC. The rest of the coefficients are considered spurious and shrunk to zero. It is informative to note that the link between zero PSCs and zero AR coefficients is not exact. We give more details on this relationship in Appendix~\ref{s:arpsc}.

In this paper, we improve the  \citet{Davis2016} framework, by proposing a modification of their two-stage sVAR method, calling it modified sVAR (msVAR), with the following two key distinctions:

1. In msVAR, zeros of PSC are identified by employing time series graphical lasso (TSGlasso) \citep{Jung2015, foti2016, tugnait2018} to estimate the inverse spectral density matrix. The main advantage of such modification is to avoid inversion of a possibly  high-dimensional matrix. However, TSGlasso involves optimization over Hermitian matrices with complex-valued entries which needs a special treatment, see Appendix~\ref{a:wc}.

2. We use  FDR in the refinement stage. The impetus of the FDR utilization  is to substitute many pairwise hypothesis tests with a multiple hypothesis testing, which provides a better model selection framework \citep{Benjamini2009, Barber2015}. In Section~\ref{a:compst1}, our  simulation results show the advantages of the FDR refinement in the second stage.

The remainder of the paper is organized as follows. Section~\ref{s:s2} introduces details on the multivariate time series analysis, VAR, and the two-stage sVAR method. Time series graphical models and TSGlasso are discussed in Section~\ref{s:s3},
where TSGlasso requires tuning parameter selection to control sparseness and  similarity of undirected graphs corresponding to the inverse spectral density matrices across the Fourier frequencies. Section~\ref{s:s4}  provides details of our algorithm for msVAR.
In Section~\ref{s:s5}, we study and compare sVAR and msVAR using simulated and a real datasets. For the real data, in addition to sVAR and msVAR models, we consider Bayesian Ridge Regression VAR (BRRVAR) \citep{Banbura}, VAR with Lasso (LASSOVAR) penalty \citep{Song2011}, and VAR with hierarchical componentwise (HVARC) and Own/Other (HVAROO) \citep{Nicholson2016} penalties. 
Finally, we conclude with the discussion in Section~\ref{s:s6}. 

\section{ Multivariate Stationary VAR Models} \label{s:s2}

In this section, we review  some  basic properties of multivariate stationary processes, their
spectral density matrices, VAR, and sVAR models.

\subsection{Partial Spectral Coherence} \label{s:psc}

 In this part,  we  give a brief introduction to the PSC estimation. A deeper treatment can be found in \citet{Brillinger2001, Brockwell1986}.
Let $\{Y_{t,i}\}$ and $\{Y_{t,j}\}$ be two distinct marginal series of a $K$-variate stationary process $\{Y_t\}$, and $\{ Y_{t,-ij}\}$ denotes the remaining $(K-2)$ marginal processes. The conditional correlation
 between two time series is computed by adjusting for the linear effect
 of the remaining marginal series $\{Y_{t,-ij}\}$. The linear effect of $\{Y_{t,-ij}\}$ is removed from the $\{Y_{t,i}\}$ by determining the optimal
filter $\{D_{k,j}\}$ \citep{Dahlhaus2000}, which minimizes $E(Y_{t,i} - \sum_{k=-\infty}^{\infty}D_{ki}Y_{t-k,-ij})^2.$
Then, the residual after removing the linear filter is:
\begin{equation} \label{eq:epsil}
\epsilon_{t,i}:= Y_{t,i}- \sum_{k=-\infty}^{\infty}D_{k,i}Y_{t-k,-ij}.
\end{equation}

The residual $\epsilon_{t,j}$ after removing the linear effect from $Y_{t,j}$ can be estimated similarly. Thus, 
two marginal series $\{Y_{t,i}\}$ and $\{Y_{t,j}\}$
are conditionally uncorrelated, if and only if their residual series  $\{\epsilon_{t,i}\}$
and $\{\epsilon_{t,j}\}$ are uncorrelated at all lags; i.e.,
$\mbox{cor}(\epsilon_{t+k,i},\epsilon_{t,j})=0$ for all $k \in \mathbb{Z}$. In the frequency domain, $\{\epsilon_{t,i}\}$ and
$\{\epsilon_{t,j}\}$ series being uncorrelated at all lags is  equivalent to the cross-spectral density
of  two residual series being zero at all normalized frequencies
$\omega \in[0, 1]$, and the residual (cross)spectral density is defined as the Fourier transform of the autocovariance sequence 
\begin{equation}\label{eq:spdensity}
f_{ij}^{\epsilon}(\omega): = \frac{1}{2\pi}\sum_{k=-\infty}^{\infty}{c_{ij}(k)^{\epsilon}}e^{(-i2\pi k\omega)},\; \omega \in [0,1],
\end{equation}
where $c_{ij}(k)^{\epsilon}=cov(\epsilon_{t+k,i},\epsilon_{t,j})$ is the (cross)autocovariance function of
 two marginal processes in which we tacitly assume $\sum|c_{ij}(k)|<\infty$.
The PSC between two distinct marginal series \citep{Brillinger2001, Brockwell1986} is defined as:
\begin{equation}\label{eq:PSC}
PSC_{ij}(\omega):=\frac{f_{ij}^{\epsilon}(\omega)}{\sqrt{f_{ii}^{\epsilon}(\omega)f_{jj}^{\epsilon}(\omega)}}, \; \omega \in[0,1].
\end{equation}

The residual cross-spectral density  $f_{ij}^{\epsilon}(\omega)$
 can be computed from the spectral density $f^{Y}(\omega)$ of the
process $\{Y_t\}$ by
\begin{equation} \label{eq:spd}
f_{ij}^{\epsilon}(\omega) = f_{ij}^{Y}(\omega) - f_{i,-ij}^{Y}(\omega)(f_{-ij,-ij}^{Y}(\omega))^{-1} f_{-ij,j}^{Y}(\omega),
\end{equation}
where $f_{i,-ij}^{Y},\,f_{-ij,-ij}^{Y}$ and $f_{-ij,j}^{Y}$ are some partitions of the spectral density matrix.
(\ref{eq:spd}) involves inverting $(K-2)\times (K-2)$  matrices $f_{-ij,-ij}^{Y}(\omega)$  for $\binom{K}{2}$ pairs,
which is computationally challenging for high dimensional data. \citet{Dahlhaus2000} proposed an efficient method to simultaneously compute PSC
for all
$\binom{K}{2}$ pairs by inverting the $K\times K$ spectral density matrix. Thus, setting $\Theta^{Y}(\omega):= f_{ij}^{Y}(\omega)^{-1}$, the
 PSC can be computed as follows:
\begin{equation}\label{eq:invpsc}
PSC_{ij}(\omega)= -\frac{\Theta_{ij}^{Y}(\omega)}{\sqrt{\Theta_{ii}^{Y}(\omega)\Theta_{jj}^{Y}(\omega)}},
\end{equation}
where $\Theta_{ii}^{Y}(\omega),\, \Theta_{jj}^{Y}(\omega)$  are the $i$th and $j$th diagonal entries and $\Theta_{ij}^{Y}(\omega)$
the $(i,j)$th entry of $\Theta^{Y}(\omega)$. From
(\ref{eq:spdensity}),(\ref{eq:PSC}) and  (\ref{eq:invpsc}), it can be seen that $\{Y_{t,i}\}$ and $\{Y_{t,j}\}$ are conditionally uncorrelated iff
$\Theta_{ij}^{Y}(\omega)=0$ for all $\omega \in [0,1]$. Note that PSC and the inverse spectral density matrix are analogs of familiar notions of partial correlations and inverse covariance matrices in multivariate statistics.

\subsection{VAR and sVAR Models}

 Consider $K$ dimensional VAR model of order p (VAR(p)):
\begin{equation} \label{eq:2.1}
Y_t = a + A_{1}Y_{t-1} + \dots + A_{p}Y_{t-p} + u_t,
\end{equation}
where $Y_t =(y_{1t},\dots,y_{Kt})^{T}$ is a random vector, $A_i$'s are fixed $(K\times K)$
 coefficient matrices, and $a =(a_{1},\dots,a_{K})^T$
is a $(K\times1)$ vector of intercept. The \textit{K}-dimensional white noise is given
by $u_t = (u_{1t},\dots, u_{Kt})^T$ where $E(u_{t})= 0$, $E(u_{t}u_{s}^T) = \Sigma_{u}$
 for $t=s$, and   $E(u_{t}u_{s}^{T}) = 0$ otherwise.
 We further assume that the process is stable, i.e., $\mbox{det}(I_k - \sum_{i = 1}^p A_iz^i) \neq 0$ for $z \in \mathcal{C}, \, |z| \leq 1$.

 Given time series observations $Y_1,\ldots,Y_T$, fitting VAR models amounts to estimating  the lag order $p$ and the coefficient
 matrices $A_1,\ldots,A_p$. However, when $K$ is large or even moderate, the VAR model is over-parametrized since the number of parameters grows quadratically ($K^2p$). Therefore, there is growing interest in developing sparse methods to overcome the computational problem and the interpretation of the model parameters, see for example \citet{Song2011, Davis2016, Nicholson2016, ding2017, yuen2018, safikhani2020}.

\subsubsection{Two-stage sVAR}

In this section, we describe the two-stage sVAR approach introduced in \citet{Davis2016}. Algorithm~\ref{tb:svartable} reports Stage 1 and 2 of the sVAR algorithm. 

\begin{table}[ht!]
\centering
\captionsetup{name=Algorithm}
\caption{sVAR algorithm}
\label{tb:svartable}
\begin{tabular} {p{15cm}}
\hline
\textbf{Stage 1}
\begin{enumerate}[noitemsep]
	\item[1.] Invert estimated spectral density matrix and compute the PSC for all $K(K-1)/2$ pairs of distinct marginal series.
	\item[2.] Construct a sequence $Q_1$ by ranking summary statistics $\mathcal{S}_{ij}$'s (see (\ref{eq:Sij})) from highest to lowest.
	\item[3.] For each $(p,M)\in \mathbb{P \times M}$, set the order of autoregression to $p$ and feed the top $M$ pairs in the
sequence $Q_1$ to the VAR model. 		
Estimate parameters under this constraint and compute the corresponding $BIC(p,M)$:
	\begin{equation}
	BIC(p,M) = -2 \log L(A_1,\dots,A_p) + \log T \cdot(K+2M)p
	\end{equation}
	\item[4.] Choosethe number of lags and non-zero pairs $(p^*,M^*)$ that gives the minimum BIC value over $\mathbb{P \times M}$
\end{enumerate}
\textbf{Stage 2}
\begin{enumerate}[noitemsep]
	\item[1.] For each of the non-zero AR coefficient compute the t-statistic via (\ref{eq:tstat}),
	\item[2.] Construct the sequence $Q_2$ of the $(K+2M^*)p^*$ triplets $(i,j,k)$ by ranking $|t_{i,j,k}|$ from highest to lowest,
	\item[3.] For each $m \in \{0,1,\dots,(K+2M^*)p^*\}$ select the $m$ non-zero AR coefficient corresponding to the top triplets in the sequence $Q_2$ and
compute $\mbox{BIC}(m)=-2\log L + \log T \cdot m$,
	\item[4.] Choose the number of non-zero $m^*$ that gives the minimum BIC value: $\mbox{BIC}(m^*)=-2\log L + \log T \cdot m^*$.
\end{enumerate}
\\
\hline
\end{tabular}
\end{table}

\paragraph{Stage 1: Model Selection:}
The first stage exploits PSC to set to zero certain entries of coefficient matrices. More precisely,
  \begin{equation} \label{eq:assert}
  PSC_{ij}(\omega)=0, \; \omega \in [0, 1] \Longrightarrow A_k(i,j)=A_k(j,i)=0, \; k=1,\dots,p .
  \end{equation}
   As discussed, this relationship is only an assertion and is not exact for some cases. See Section~\ref{s:arpsc} for more details.

 Thus, a group of AR coefficient estimates is set to zero if the corresponding PSC estimates are zero.
However, because of the sampling variability, estimated PSCs
 are not exactly zero even though two marginal series
are conditionally uncorrelated. \citet{Davis2016} overcome this problem by ranking estimated PSCs from largest to smallest and finding a threshold that separates non-zero PSCs in which the supremum of the squared modulus of the estimated PSC is used as summary statistics:
\begin{equation} \label{eq:Sij}
\mathcal{S}_{ij}:=\sup|PSC_{ij}(\omega)|^2,
\end{equation}
where the supremum is taken over all scaled frequencies $\omega$. Thus, a large value of $\mathcal{S}_{ij}$
indicates that two marginal series are conditionally correlated and vice versa.
%
%

The output of the first stage algorithm is a model with $(K+ 2M^*)p^*$ non-zero AR coefficients. If
the proportion of selected pairs is small, then the number of non-zero parameters is much
smaller than that for the fully-parametrized VAR$(p^*)$, where the number of parameters is $K^2p^*$. In Step 3, the parameter estimation under the constraint is implemented using constrained log-likelihood estimation described in \citet[Chapter 5]{Lutkepohl2007}.

\paragraph{Stage 2: Refinement:}
In Stage 1, the PSC can only be evaluated for pairs of series, but it does not consider diagonal
entries in $A_1,\dots,A_p$ and  within the group
coefficients for each pair of component series. In other words, Stage 1 may produce
 spurious non-zero AR coefficients. To emancipate the model from spurious
coefficients, in the $(K+2M^*)p^*$ sequence of non-zero AR coefficients, \citet{Davis2016} rank them according to the absolute value
 of their t-statistics, i.e.,
\begin{equation} \label{eq:tstat}
t_{i,j,k}:= \frac{A_k(i,j)}{se(A_k(i,j))},
\end{equation}
where the standard error $\mbox{se}(A_k(i,j))$ is computed from the asymptotic distribution of the constrained maximum likelihood
 estimator \citep{Lutkepohl2007}.
%
%
%
After the second stage, the procedure leads to a sparse VAR model that contains $m^*$ non-zero AR coefficients, denoted by $\mbox{sVAR}(p^*,m^*)$.

\section{Time Series Graphical Lasso } \label{s:s3}

The salient feature of Gaussian graphical models is to represent conditional independencies among random variables in multivariate data. An undirected graph is a powerful tool for visualizing these relationships where the vertices represent the random variables, and the edge between two vertices indicates the conditional dependence of corresponding variables. For a $K$-dimensional random vector $Y \sim N_K(0, \Sigma)$, a Gaussian graphical model can be constructed from the inverse covariance matrix $\Theta = \Sigma^{-1}$. More precisely, a zero off-diagonal entry of $\Theta_{ij } = 0$ implies that $Y_i$ and $Y_j$ are conditionally independent given all other variables \citep{Whittaker1990}. When $K$ is large, it is reasonable to impose structure or regularize $\Theta$ directly in the search for sparsity \citep{Banerjee2007, Friedman2008}, see \citet{pourahmadi2013} for an overview.

Given the sample data, a regularized  Gaussian graphical model estimation can be formulated as

 \begin{equation} \label{eq: glasso}
 \min_{\Theta} \log \mbox{det}(\Theta) - \mbox{tr}(S\Theta) + P(\Theta, \lambda),
 \end{equation}
where $\mbox{det}(\cdot)$ is determinant of the matrix, $S$ is the sample covariance matrix, $P(\cdot)$ is a penalization term, and $\lambda$ is a tuning parameter. In \citet{Banerjee2007, Friedman2008} the penalization term is $\ell_1$ norm, i.e., $P(\Theta,
\lambda) = \lambda \|\Theta\|_1 = \lambda \sum_{ij}|\theta_{ij}|$. The  Glasso algorithm is extremely popular and has been  extended to  multiple covariance matrices in \citet{Guo2011,Danaher2014} where the
 data from several populations may have a similar graphical structure.

\citet{brillinger1996} and \citet{Dahlhaus2000} have extended the use of graphical models to the multivariate time series setup. Consider $K$ dimensional stationary process $Y_t$, $t = 1, \dots, T$. Let $G = (V, E)$ denote a graph, where each node $v \in V$ corresponds to one of the times series in $Y_t$ and the edge between nodes is characterized by the conditional dependence of the marginal series $Y_i$ and $Y_j$, given the rest $Y_{-ij}$; i.e., $Y_i \mathdash Y_j$ iff $\Theta^{Y}_{ij}(\omega) = 0$, or $\mbox{PSC}(\omega) = 0$ for  $\omega \in [0, 1]$. From now on, whenever there is no confusion in the context, we drop the superscript $Y$ from $\Theta$, and note that $\Theta(\cdot)$ is a Hermitian matrix-valued function with complex-valued entries.


A time series extension of the Glasso requires expressing the log-likelihood function in terms of the discrete Fourier transform of the data and $\Theta(\cdot)$. Define the normalized discrete Fourier transform (DFT) of $K$ dimensional random vector $Y_t$,
\begin{equation}
	d(\omega_n) = \frac{1}{\sqrt{T}} \sum_{t = 0}^{T - 1} Y_t \mbox{exp}(-i2\pi\omega_n t),
\end{equation}
where $i = \sqrt{-1}$, $\omega_n = n / T, \; n = 0,1, \dots, T - 1$. Since $Y_t$ is real-valued, the complex conjugate $d^*(\omega_n) = d(-\omega_n) = d(1 - \omega_n)$ and for $n = 0,1, \dots, T/2$, $d(\omega_n)$ is completely determined for all $n$. Moreover, from \citet[Theorem 4.4.1]{Brillinger2001} as $T \rightarrow \infty$, $d(\omega_n),\; n = 1,2, \dots, (T/2) - 1$ are independent complex Gaussian $N_c(0, f[\omega_n])$ random vectors and for $n = \{0, T/2\}$, $d(\omega_n)$ are independent real Gaussian $N_r(0, f[\omega_n])$. Ignoring $n = \{0, T/2\}$ frequency points, and denoting $f[\omega_n] = f[n],\; \Theta[n] = f^{-1}[n]$,  the joint pdf for $d(\omega_n),\, n = 1, \dots, (T/2) - 1$ is

\begin{equation}\label{eq:pdf}
g(d(\omega_1),\dots, d(\omega_{(T/2) - 1})) = \prod_{n = 1}^{(T/2) - 1} \frac{\mbox{exp}(-d^H(\omega_n)\Theta[n]d(\omega_n))}{\pi^K \mbox{det}(f[n])}
\end{equation}

A standard assumption in spectral density estimation is locally smoothness \citep{Brillinger2001, stoica1997}, i.e., $f[n]$ is approximately constant over $L = 2m_t + 1 \geq K$ consecutive frequency points where $m_t$ is the half-window size. After carefully picking

\[\tilde \omega_l = \frac{(l - 1)L + m_t + 1}{T};\; M =  \Big \lfloor \frac{T/2 - m_t - 1}{L} \Big \rfloor;\; l = 1,2, \dots M,\]

leads to $M$ equally spaced frequencies $\tilde \omega_l$. Therefore, the exploitation of the local smoothness assumption results for $k = -m_t, -m_t + 1, \dots, m_t$
\begin{equation} \label{eq:smooth}
\tilde \omega_{l,k} = \frac{(l - 1)L + m_t + 1 + k}{T};\, f[k] = f[\{l,k\}].
\end{equation}
From (\ref{eq:smooth}) and (\ref{eq:pdf}), the pdf is
\begin{equation} \label{eq:modpdf}
\begin{aligned}
g(d(\omega_1),\dots, d(\omega_{(T/2) - 1})) &= \prod_{n = 1}^{M} \prod_{l = -m_t}^{m_t}  \frac{\mbox{exp}(-d^H(\tilde \omega_{l,n})\Theta[n]d(\omega_{l,n}))}{\pi^K \mbox{det}(f[n])^L} \\
&= \prod_{l = 1}^{M} \frac{\mbox{exp} [-\mbox{tr}(\tilde f[n] \Theta[n] )]}{\pi^{L K} \log \mbox{det} f[n]},
\end{aligned}
\end{equation}
where $\tilde f[n] = \sum_{l= -m_t}^{m_t} d(\tilde \omega_{l,n})d^H(\tilde \omega_{l,n}) / L$ is the sample spectral density matrix whose entries are potentially complex-valued. Thus, the log-likelihood function can be written as

\[ W(\Theta[\cdot]) = \sum_{n = 1}^{M} L \left[\log \mbox{det} (\Theta[n]) - \mbox{tr}(\tilde f[n] \Theta[n]) \right].\]
Analogous to Glasso, we introduce sparsity by minimizing the following regularized log-likelihood
\begin{equation} \label{eq:tslog}
\min_{\Theta[n]} W(\Theta[\cdot]) +   P(\Theta[\cdot], \lambda),
\end{equation}
where
\begin{equation} \label{tsggl:1}
P(\Theta[\cdot], \lambda) =  \lambda \sum_{i\neq j} \sqrt{\sum_{n=1}^{N} |\Theta_{ij}[n]|^2 }, \; \Theta[\cdot] = \{\Theta[1], \dots, \Theta[N]\},
\end{equation}

 As in \citet{Jung2015}, we appeal to the alternating direction method of multipliers (ADMM)  \citep{Boyd2011} for minimization. However, in the
 following we pay due attention to the fact that the entries of $\Theta$ are complex-valued. The ADMM minimizes the scaled augmented Lagrangian 
\begin{equation} \label{eq:14}
\begin{split}
L_\rho(\Theta[\cdot],Z[\cdot],U[\cdot])&= \sum_{n=1}^{M} L \left[ -\log \mbox{det}(\Theta[n]) + \mbox{tr}(\tilde{f}[n] \Theta[n])) \right] + \lambda P(Z[\cdot]) \\
& +{\rho / 2}\sum_{n =1}^{M} \|\Theta[n]-Z[n]+U[n]\|_F^2,
\end{split}
\end{equation} subject to $\Theta[n] = Z[n]$ and  $\|X[n]\|^2_F = \sum_{i j} |X_{ij}[n]|^2$ for $n =1,\dots, M$.
Given $(\Theta^{(k)}[\cdot],Z^{(k)}[\cdot],U^{(k)}[\cdot])$  matrices in the $k$th iteration, the ADMM algorithm implements the following three updates for the next ($k +1 $) iteration:
\begin{enumerate}
	\item[(a)]\label{step1}  
$	\Theta^{(k+1)}[\cdot] \leftarrow \argmin_{\Theta[\cdot]} L_\rho(\Theta[\cdot],Z^{(k)}[\cdot],U^{(k)}[\cdot])$
	\item[(b)] \label{step2}
	$Z^{(k+1)}[\cdot] \leftarrow  \argmin_{Z[\cdot]}  L_\rho(\Theta^{(k+1)}[\cdot],Z[\cdot],U^{(k)}[\cdot])$
 		\item[(c)] \label{step3} $U^{(k+1)}[\cdot] \leftarrow U^{(k)}[\cdot]+(\Theta^{(k+1)}[\cdot]- Z^{(k + 1)}[\cdot])$
\end{enumerate}
It is timely and instructive to note that unlike the formulation in (\ref{eq: glasso}) \citep{Banerjee2007, Friedman2008, Danaher2014} where $\Theta$ is real-valued, here we have to deal with a complex-valued $\Theta[\cdot]$ in (\ref{eq:14}). While \citet{li2015} establish steps for the complex-valued ADMM when the penalty function is $\ell_1$ norm, here we resort to  Wirtinger calculus \citep{wirtinger1927,brandwood1983}, coupled with the definition of Wirtinger subgradients \citep{bouboulis2012},  to derive updates (a)-(c) for matrices with complex entries. Details are relegated to Appendix~\ref{a:wc}.

%
%
\citet[Section 3.1]{Boyd2011} showed that for a given $\rho$, the convergence of iterates to the global minimum is guaranteed. The choice of $\rho$ controls the speed of convergence. \citet[Section 3.4.1]{Boyd2011} discuss the adaptive choice of $\rho$ to improve convergence. On the statistical side, \citet{Jung2015} provide an upper bound on the support recovery of the TSGlasso.

\subsection{Tuning Parameter Selection for TSGlasso} \label{s:tune}
In the TSGlasso algorithm, the tuning parameter $\lambda$ controls the sparsity and the similarity of the estimated undirected graphs over the scaled frequencies.
 In this section, we review some classical methods for tuning parameter selection. 

Ideally, the selected tuning parameter should produce an undirected graph that is sufficiently complex to be interesting, sufficiently sparse to be interpretable, and, more importantly, should be supported by data. The traditional approaches such as the Akaike information criterion (AIC), Bayesian information criterion (BIC) and cross-validation tend to choose graph that is too large \citep{meinshausen2010}. \citet{homrighausen2018} empirically showed that for the penalized regression, the tuning parameter, selected from the Stein unbiased risk estimator (SURE)-type criterion, tends to perform better than other considered criteria. For graphical models, extended BIC, introduced in \citet{foygel2010}, shows practical superiority compare to discussed criteria. 

 From (\ref{eq:modpdf}), the AIC approximation for the time series graphical model is:
\begin{equation} \label{par:1}
\mbox{AIC}(\lambda) = \sum_{n=1}^{M} \left[ - \log \mbox{det}\hat \Theta_{\lambda}[n] + \mbox{tr}(\tilde{f}(n) \hat \Theta_{\lambda}[n]) \right] \times L + 2E_n,
\end{equation}
 where $\hat \Theta_{\lambda}[n]$ is the estimated inverse spectral density at tuning parameter $\lambda$, and $E_n$ is the number of non-zero elements in $\hat \Theta_{\lambda}[n]$. Using AIC, we choose $\lambda$ which gives the minimum value of (\ref{par:1}). Similarly, an approximation of the extended BIC is:
\begin{equation} \label{par:2}
\mbox{eBIC}(\lambda) = \sum_{n=1}^{M} \left[ - \log|\hat \Theta_{\lambda}[n]| + tr(\tilde{f}[n] \hat \Theta_{\lambda}[n]) \right] \times L + \log(L) E_n + 4 E_n \gamma \log(K),
\end{equation}
with a hyper-parameter $\gamma \in[0,1]$. If $\gamma =0$, then (\ref{par:2}) reduces to the classical BIC. The higher value of $\gamma$ leads to the stronger penalization of large graphs. For the moderate and large values of $K$, \citet{foygel2010} suggest $\gamma = 0.5$.

\section{Modified  sVAR} \label{s:s4}

In this section, we introduce our msVAR procedure and highlight the key differences with the sVAR. Algorithm~\ref{tb:msvartable} summarizes the proposed modifications. 

\begin{table}[ht!]
\captionsetup{name=Algorithm}
\centering
\caption{msVAR algorithm}
\label{tb:msvartable}
\begin{tabular} {p{15cm}}
\hline
\textbf{Stage 1}
\begin{enumerate}[noitemsep]
	\item[1.] Estimate the inverse spectral density matrix using  TSGlasso, and let  $M^*$ be its number of non-zero elements (see Section~\ref{s:s3}).
	\item[2.] Estimate the AR parameters under the zero constraint and choose the number of lags ($p^*$) by minimizing
	\[BIC(p) = - 2 \log L(A_1, \dots, A_p) + \log T(K + 2M^*)p\]
\end{enumerate}
\textbf{Stage 2}
\begin{enumerate}[noitemsep]
	\item[1.] For each of the non-zero AR coefficient compute the t-statistic and p-value.
	\item[2.] Choose $m^*$ non-zero coefficients that reject hypothesis in FDR procedure with the threshold value of FDR-corrected significance $q$.
\end{enumerate}
\\
\hline
\end{tabular}
\end{table}

The first stage of our Algorithm is designed to avoid the costly matrix inversion and grid search procedure to compute constrained MLE of the AR parameters. We substitute Steps 1 - 3 of the sVAR's Stage 1 (see Algorithm~\ref{tb:svartable}) by the TSGlasso algorithm, with the ensuing $M^*$ non-zero elements, and use BIC only once  to choose the number of lags $p^*$ compare to sVAR (see Step 3 in Stage 1 in Algorithm~\ref{tb:svartable}).



In the second stage,  instead of using  t-statistics of the AR coefficients, our algorithm relies on the FDR \citep{Benjamini1995} procedure for further refinement. There is a  rich literature on the use of FDR for model selection, for example, see \citet{Benjamini2009, Barber2015, GSell2016}, etc.  The advantage of FDR utilization in the second stage is twofold: First, instead of pairwise t-statistics, we implement multiple hypothesis testing. Second, it eliminates the need for Step 3 in Stage two of sVAR (see Algorithm~\ref{tb:svartable}). The empirical analysis in Appendix~\ref{a:compst1} further shows the advantages of using FDR for refinement.

\section{Numerical Results} \label{s:s5}
In this section, we use simulated and real datasets to compare the performance of msVAR and sVAR models. Our analyses indicate that msVAR is a preferred choice when the goal is to learn the structure of the coefficient matrix. On the other hand, sVAR outperforms msVAR when the ultimate task is forecasting. The analysis in Section~\ref{s:rtime} shows that the proposed modifications significantly improve the computation time of the algorithm.

\subsection{Comparing the msVAR and sVAR Models}
\subsubsection{Evaluation Measures and Visualization}
	
The following metrics are computed to compare the performance of the two methods:
\begin{itemize}[noitemsep]
	\item the squared bias of the AR coefficient  estimates:
	\[\sum\limits_{k=1}^{max(p,\hat{p})}\sum\limits_{i,j=1}^{K}[\mathbb{E}[\hat{A_k}(i,j)]-A_k(i,j)]^{2};\]
	\item the variance of the estimated AR coefficient:\[\sum\limits_{k=1}^{max(p,\hat{p})}\sum\limits_{i,j=1}^{K}var(\hat{A_k}(i,j));\]
	\item the mean square error (MSE) of the AR coefficient estimates:
    \[\sum\limits_{k=1}^{max(p,\hat{p})}\sum\limits_{i,j=1}^{K}\{[\mathbb{E}[\hat{A_k}(i,j)]-A_k(i,j)]^{2} +var(\hat{A_k}(i,j))\}.\]
    \item the true positive rate (TPR): estimates the ratio between the number of correctly found edges in estimated graph and the number of true edges in the true graph.
    \item the false positive rate (FPR): estimates  the ratio between the number of incorrectly found edges in estimated graph and the number of true missing edges in the true graph.
\end{itemize}

In addition, we utilize a \citet{eichler2012}'s  proposal to visualize an estimated VAR model using a mixed graph. The edge set $E_m$ of the mixed graph $G_m = (V_m, E_m)$ consists of directed and undirected edges, in which

\begin{itemize}[noitemsep]
	\item $i \rightarrow j \notin E_m$ whenever $A_k(i,j) = 0,\, k = 1, \dots, p$
	\item $i \mathdash j \notin E_m$ whenever $(\Theta_u)_{ij} = (\Theta_{u})_{ji} = 0.$
\end{itemize}

In other words, the directed edge is in the edge set whenever $Y_i$ is Granger-causal for $Y_j$ \citep{Lutkepohl2007}, and an undirected edge is in the edge set whenever $Y_i$ and $Y_j$ are
contemporaneously conditionally dependent. However, for the sake of clarity, we present only the directed part of the mixed graph as in Figures~\ref{fig:m1}-\ref{fig:m3}.

\subsubsection{The Simulation Setup}
	
In the simulation study, we consider three different \textit{stable}  VAR models to compare performance of sVAR and msVAR methods.

\begin{itemize}

\item[Model 1:] $y_t = A_1 y_{t-1} + u_t$, and $K = 10$ \\

\[
\scriptsize{
A_1 = \begin{bmatrix}[0.5]
 $0$ & $0$ & $0$ & $0$ & $0$ & $0$ & $0$ & $0.3$ & $0$ & $0$ \\
$0$ & $0$ & $0.1$ & $0$ & $0$ & $0$ & $0.4$ & $0$ & $0$ & $0.4$ \\
$0$ & $0.6$ & $0$ & $0$ & $0$ & $0$ & $0$ & $0$ & $0$ & $0$ \\
$0$ & $0.2$ & $0$ & $0$ & $0.5$ & $0$ & $0$ & $0$ & $0$ & $0$ \\
$0$ & $0.3$ & $0$ & $0.1$ & $0$ & $0$ & $0.2$ & $0.1$ & $0.3$ & $0.5$ \\
$0.2$ & $0$ & $0$ & $0$ & $0.4$ & $0$ & $0$ & $0$ & $0$ & $0$ \\
$0$ & $0$ & $0$ & $0$ & $0$ & $0$ & $0$ & $0$ & $0$ & $0.6$ \\
$0$ & $0$ & $0$ & $0$ & $0$ & $0.6$ & $0$ & $0$ & $0$ & $0$ \\
$0.2$ & $0$ & $0$ & $0$ & $0$ & $0$ & $0$ & $0$ & $0.2$ & $0$ \\
$0$ & $0$ & $0$ & $0$ & $0.4$ & $0$ & $0$ & $0$ & $0$ & $0$ \\
	\end{bmatrix},
\Theta_{u} = \begin{bmatrix}[0.5]
\delta & \delta /2 &\dots & \delta /10 \\
\delta / 2 & $1$ & \hdots & $0$ \\
\vdots & \vdots & \ddots & $0$ \\
\delta / 10 & $0$ & \hdots & $1$  \\
	\end{bmatrix}}
	,\]
	where $\delta = 0.5$ .
The setup of the simulation is borrowed from the \citet[Section 4]{Davis2016}.
See Figure~\ref{fig:m1} for illustration. 

\item[Model 2:] $y_t = A_1 y_{t-1} + u_t$, and $K = 6$ \\
\[
\scriptsize{
A_1 =
	\begin{bmatrix}[0.5]
	$0$ & $0.50$ & $0.50$ & $0.20$ & $0$ & $0$ \\
	$0$ & $0$ & $0.30$ & $0$ & $0$ & $0$ \\
	$0$ & $0.25$ & $0.50$ & $0$ & $0$ & $0$ \\
	$0$ & $0$ & $0$ & $0$ & $0.33$ & $0.33$ \\
	$0$ & $0$ & $0$ & $0$ & $0$ & $0.20$ \\
	$0$ & $0.50$ & $0$ & $0$ & $0.17$ & $0.33$ \\
	\end{bmatrix},	
\Theta_u =
	\begin{bmatrix}[0.5]
	$0.17$ & $0$ & $0.25$ & $0.030$ & $0$ & $0$ \\
	$0$ & $1.40$ & $0.34$ & $0.25$ & $0.04$ & $0.58$ \\
	$0.25$ & $0.34$ & $0.55$ & $0.05$ & $0$ & $0$ \\
	$0.03$ & $0.25$ & $0.05$ & $0.26$ & $0$ & $0.42$ \\
	$0$ & $0.04$ & $0$ & $0$ & $1.51$ & $0.36$ \\
	$0$ & $0.58$ & $0$ & $0.42$ & $0.36$ & $0.98$ \\
	\end{bmatrix}}
\]

The purpose of this setup is to compare methods when for some entries of $A_1$ the assertion~(\ref{eq:assert}) is violated
\begin{equation} \label{eq:cnd2}
	\mbox{PSC}_{ij}(\omega) = 0,\, \omega \in [0,1]\; \mbox{and}\;  A_1[i, j] \neq 0\;
\end{equation}
See Figure~\ref{fig:m2} for illustration.  
Red directed edges indicate entries of $A_1$ that violate assertion~(\ref{eq:assert}).
\begin{remark}
In general,  for any element of $A_1$, enforcement of (\ref{eq:cnd2}) is not a trivial task. Values and structure of coefficient matrices $A_1$ and $\Theta_u$ are carefully constructed such that $A_1[1,2]$ and $A_1[4,5]$ satisfy condition (\ref{eq:cnd}) located in the Appendix.
\end{remark}

\item[Model 3:] $y_t = A_1 y_{t-1} + A_2 y_{t-2} + u_t$, and $K = 6$ \\

\[
\scriptsize{
A_1 =
	\begin{bmatrix}[0.5]
		$$-$0.6$ & $0.4$ & $0$ & $0$ & $0$ & $0.4$ \\
		$0.4$ & $$-$0.6$ & $0.4$ & $0$ & $0$ & $0$ \\
		$0$ & $0.4$ & $$-$0.6$ & $0.4$ & $0$ & $0$ \\
		$0$ & $0$ & $0.4$ & $$-$0.6$ & $0.4$ & $0$ \\
		$0$ & $0$ & $0$ & $0.4$ & $$-$0.6$ & $0.4$ \\
		$0.4$ & $0$ & $0$ & $0$ & $0.4$ & $$-$0.6$ \\
	\end{bmatrix},
A_2 =
	\begin{bmatrix}[0.5]
		$$-$0.3$ & $0.2$ & $0$ & $0$ & $0$ & $0.2$ \\
		$0.2$ & $$-$0.3$ & $0.2$ & $0$ & $0$ & $0$ \\
		$0$ & $0.2$ & $$-$0.3$ & $0.2$ & $0$ & $0$ \\
		$0$ & $0$ & $0.2$ & $$-$0.3$ & $0.2$ & $0$ \\
		$0$ & $0$ & $0$ & $0.2$ & $$-$0.3$ & $0.2$ \\
		$0.2$ & $0$ & $0$ & $0$ & $0.2$ & $$-$0.3$ \\
	\end{bmatrix}}
\]

\[	
\scriptsize{
\Theta_u =
	\begin{bmatrix}[0.5]
		1 & -0.3 & 0 & 0 & 0 & -0.3 \\
		-0.3 & 1 & -0.3 & 0 & 0 & 0 \\
		0 & -0.3 & 1 & -0.3 & 0 & 0 \\
		0 & 0 & -0.3 & 1 & -0.3 & 0 \\
		0 & 0 & 0 & -0.3 & 1 & -0.3 \\
		-0.3 & 0 & 0 & 0 & -0.3 & 1 \\ 	
	\end{bmatrix}}
	\]

This setup is borrowed from \citet[Section 4]{yuen2018}. The directed graph is illustrated in Figure~\ref{fig:m3}.

\end{itemize}

For each model, the corresponding multivariate time series is generated following \citet[Appendix D1]{Lutkepohl2007} over the 50 replications. For all models, $T = 100$ and for the tuning parameter selection in Stage 1, we only report  results for the eBIC since BIC and AIC provide similar outcomes. For Stage 2 of the msVAR, the threshold value of FDR-corrected significance is fixed at $q= 0.1$.

\subsubsection{The Simulation Result}
Figures~\ref{fig:m1} - \ref{fig:m3} and Table~\ref{t:metrics} report the simulation results for Models 1 - 3, respectively. In each figure, top left directed graph corresponds to the true case and top right and bottom left to the msVAR and sVAR, respectively. The width and color shade of the estimated edges correspond to the proportion of times the edge was detected out of 50 replications; i.e., the darker and thicker the edge, more frequently  it was present and vice versa. In Figure~\ref{fig:m2}, red directed edges correspond to the condition~(\ref{eq:cnd2}). In Table~\ref{t:metrics}, for each method, the minimum of $\mbox{Bias}^2$, Variance, MSE and FPR metrics, and the maximum of TPR are highlighted.

From Figures~\ref{fig:m1} - \ref{fig:m3}, the visual comparison reveals that both msVAR and sVAR were able to detect true edges for most of the time. For Model 1, msVAR indicates better result on estimating true edges than sVAR. For example, sVAR failed to detect the edge $Y_5 \longleftrightarrow Y_{10}$ in all repetitions, while msVAR detected it around 90\% of time. The first two rows in Table~\ref{t:metrics} document the five metrics comparison for Model 1. It can be seen, that msVAR is the best for all metrics.

A similar result holds for Model 2. msVAR shows small bias but higher variance and the best TPR result. More importantly, both algorithms were able to detect the edges $Y6 \longrightarrow Y2$ and $Y1 \longrightarrow Y2$ most of the time, even thought the assertion (\ref{eq:assert}) was violated. For Model 3, the performance is reversed compared to Model 2, i.e., msVAR shows slightly higher bias but smaller variance.

\begin{figure}[htp]
\centering
\includegraphics[width=0.7\textwidth, height = 10cm]{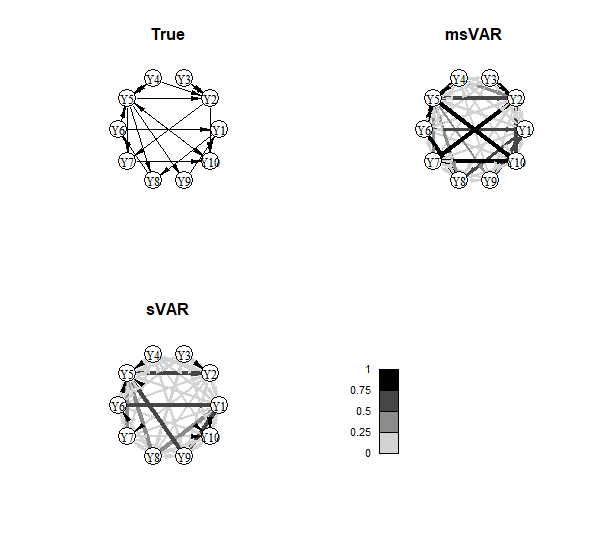} 
\caption{Model 1 simulation result. The width and color shade of estimated edges indicate the proportion of the number of times the edge was detected out of 50 replication. }
\label{fig:m1}
\end{figure}

\begin{figure}[htp]
\centering
	\includegraphics[width=0.7\textwidth, height = 9.5cm]{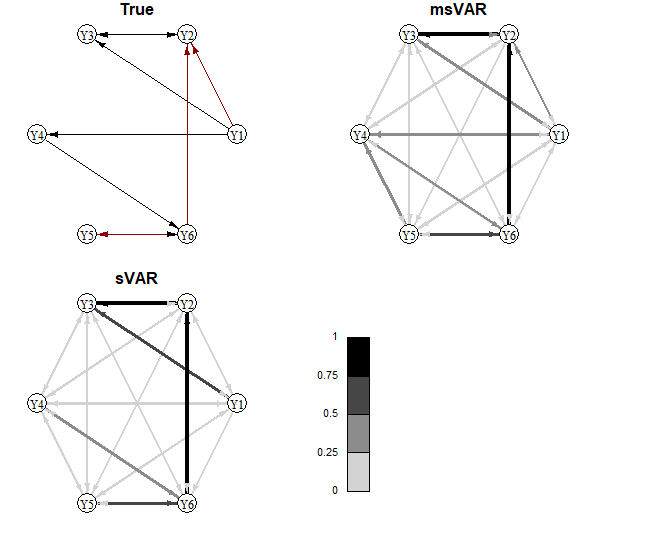} 
\caption{Model 2 simulation result. Red directed edges in the true graph correspond to the condition (\ref{eq:cnd2}). The width and color shade of estimated edges indicate the proportion of the number of times the edge was detected out of 50 replication.}
\label{fig:m2}
\end{figure}

\begin{figure}[htp]
\centering
	\includegraphics[width=0.7\textwidth, height = 9.5cm]{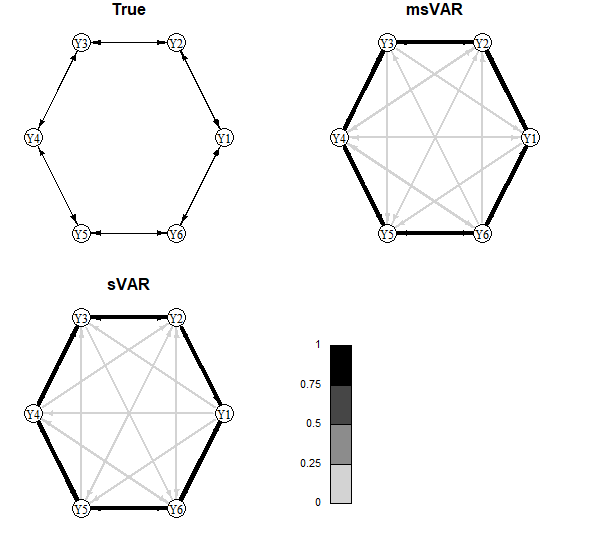} 
\caption{Model 3 simulation result. The width and color shade of estimated edges indicate the proportion of the number of times the edge was detected out of 50 repetition.}
\label{fig:m3}
\end{figure}

\begin{table}[!htbp] \centering
\setcounter{table}{0} \renewcommand{\thetable}{\arabic{table}}
  \caption{Five metrics from the sVAR and msVAR methods.}
  \label{t:metrics}
\begin{tabular}{@{\extracolsep{5pt}} ccccccc}
\\[-1.8ex]\hline
\hline \\[-1.8ex]
& Method & $\mbox{Bias}^2$ & Variance & MSE & TPR & FPR \\
\hline \\[-1.8ex]
\multirow{2}{*}{Model 1}
& sVAR & 0.178 & 0.788 & 0.966 & 0.595 & 0.03\\
&msVAR & \textbf{0.174} & \textbf{0.732} & \textbf{0.906} & \textbf{0.632}& \textbf{0.02} \\
\hline \\[-1.8ex]
\multirow{2}{*}{Model 2}
&sVAR & 0.508 & \textbf{0.609} & \textbf{1.117} & 0.477& \textbf{0.08} \\
&msVAR & \textbf{0.303} & 0.951 & 1.254 & \textbf{0.537}& 0.118\\
\hline \\[-1.8ex]
\multirow{2}{*}{Model 3}
&sVAR & \textbf{0.089} & 0.762 & 0.851 & \textbf{0.954} & \textbf{0.03}\\
&msVAR & 0.102 & \textbf{0.742} & \textbf{0.844} & 0.946 & 0.05 \\
\hline \\[-1.8ex]
\end{tabular}
\end{table}

\subsubsection{Running Time Comparison} \label{s:rtime}
In this section, we compare relative running times for sVAR and msVAR methods. For this exercise, we fixed $p  = 1$, and for $K = 15,25,50,75$, generate a sparse coefficient matrix $A$ such that the probability of having a non-zero element is equal to 0.25. Then, the coefficient matrix is rescaled to satisfy the stability condition. For the msVAR method, we select two options for time comparison: msVAR with tuning parameter selection and without, respectively. For brevity, we call those methods \textit{msVAR with} and \textit{msVAR without}.  In the former case, the tuning parameter is selected over 20 equally spaced values located in $(0,1)$ interval, and for the latter case, the tuning parameter is fixed to $\lambda = 0.2$. The relative time is reported with respect to the running time of the \textit{msVAR without}. The modified \texttt{R} code for the msVAR algorithm relies on the sVAR code framework provided in \citet{Davis2016}.  Table~\ref{t:comptime} reports relative times for $K = 15,25,50,75$.

\begin{table}[!htbp]
\centering
  \caption{Relative running times for the sVAR and msVAR algorithm. Running times are reported relative to the msVAR algorithm. }
  \label{t:comptime}
\begin{tabular}{@{\extracolsep{5pt}} cccc}
\\[-1.8ex]\hline
\hline \\[-1.8ex]
 K& msVAR with & msVAR without & sVAR \\
\hline \\[-1.8ex]
15 & $10.03$ & $1$ & $4.27$ \\
25 & $4.91$ & $1$ & $38.80$ \\
50 & $1.02$ & $1$ & $68.50$ \\
100 & $1.01$ & $1$ & $-^*$ \\
\hline \\[-1.8ex]
\footnotesize {\textit{Note:}}& \multicolumn{3}{l}{\footnotesize {$^*$ The algorithm was terminated after 24 hours.}}
\end{tabular}
\end{table}

It can be seen when $K = 15$, the \textit{msVAR without} is the fastest, followed by the sVAR, which is nearly $4.3$ times slower than the \textit{msVAR without}. Finally, the \textit{msVAR with} is almost 10 times slower than the \textit{msVAR without}. However, sVAR becomes extremely slow as $K$ grows. For large $K$, running times for the \textit{msVAR without} and \textit{msVAR with} are almost indistinguishable. The result can be explained by observing that in both methods, the computationally expensive procedure is the restricted MLE estimation, and msVAR is fast since it implements it once, compared to sVAR's grid search approach. Moreover, the computational expense of restricted MLE estimation overshadows the computation time of the tuning parameter selection as $K$ grows.

\subsubsection{Comparing msVAR stage 1 and stage 2 outputs}\label{a:compst1}
In this section, we compare msVAR stage 1 and stage 2 outputs. Recall that in stage 2 of msVAR, we use the FDR procedure for edge selection.
For comparison, we use Model 1, described in Section~\ref{s:s5}, to generate the dataset. Table~\ref{t:compst1} and Figure~\ref{fig:compst1} report the results. Results from both table and figure indicate the performance improvement after the FDR refinement in stage 2.

 \begin{table}[ht!]
 \centering
  \caption{Three metrics from the msVAR and msVAR St.1 methods.}
  \label{t:compst1}
\begin{tabular}{@{\extracolsep{5pt}} cccc}
\\[-1.8ex]\hline
\hline \\[-1.8ex]
 & Bias$^2$ & Variance & MSE \\
\hline \\[-1.8ex]
msVAR St.1 & 0.524 & 1.249 & 1.774 \\
msVAR & \textbf{0.508} & \textbf{0.925} & \textbf{1.433} \\
\hline \\[-1.8ex]
\end{tabular}
\end{table}

\begin{figure}[htp!]
\centering
\includegraphics[width=0.8\textwidth, height = 10cm]{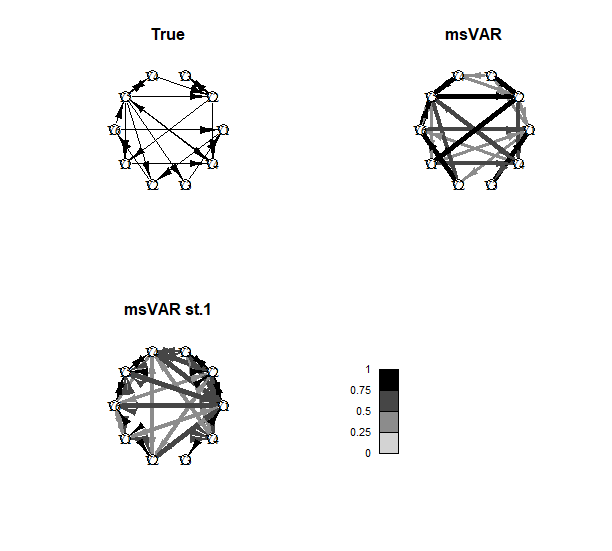} 
\caption{Model 1 simulation result. The width and color shade of estimated edges indicate the proportion of the number of times the edge was detected out of 50 replication. }
\label{fig:compst1}
\end{figure}

\subsection{Real Data Analysis}
We compare the forecasting performance of various VAR methods on a real-world macro-economic dataset. In addition to sVAR and msVAR, we consider Bayesian Ridge Regression VAR (BRRVAR) \citep{Banbura}, VAR with Lasso (LASSOVAR) penalty \citep{Song2011}, and VAR with hierarchical componentwise (HVARC) and Own/Other (HVAROO) \citep{Nicholson2016} penalties. Corresponding tuning parameters for the LASSOVAR, HVARC, and HVAROO are selected using rolling cross-validation \citep{Nicholson2016}.

 The dataset represents the 168 monthly US macro-economic time series from $01/1959$ to $02/2009$. Initially, the dataset was compiled by \citet{Stock2005} and latter augmented by \citet{koop2011}. \citet{koop2011} defines and analyzes a small (K = 3), medium (K = 20), medium-large (K = 40) and large (K = 168) VARs. In this paper, we focus only on the medium-large setup with $K = 40$ variables. To avoid a policy break, the selected sample period runs from  $01/1990$ to $2/2009$. Following \citet{koop2011}, we transform the data-set to make variables approximately stationary. For both, sVAR and msVAR, BIC  selects the number of lags to be $p = 2$. The same number of lags $p = 2$ is used for other four sparse VAR methods.

We compare the out-of-sample forecast performance for the above six VAR methods using the last 24 months ($T_{test} = 24$) as test data. We use the $h$-step-ahead forecast root mean squared error (RMSE) as a measure for the comparison.
\[\mbox{RMSE}(h) = \left[ K^{-1}(T_{test} - h + 1)^{-1}\sum_{k = 1}^K \sum_{t = T}^{T + T_{test} - h}(\hat Y_{t+h,k} - Y_{t+h,k})^2 \right]^{1/2},\]
where $\hat Y_{t+h,k}$ is the $h$-step-ahead forecast of $Y_{t+h,k}$ for $k = 1, \dots, K$.
 Table~\ref{t:rms} summarizes RMSE(h) for a forecast horizon $h = 1, 2, 3$ and $4$.  It can be seen that the sVAR perform slightly better than the msVAR for all $h = 1,\dots,4$ and HVARC is the best among six methods.

\begin{table}[!htbp] \centering
\caption{The h-step ahead forecast root mean squared error (RMSE(h)).}
\label{t:rms}
\begin{tabular}{@{\extracolsep{5pt}} lllll}
\\[-1.8ex]\hline
\hline \\[-1.8ex]
 & h = 1 & h = 2 & h = 3 & h = 4 \\
\hline \\[-1.8ex]
msVAR & 0.179 (0.07) & 0.145 (0.008) & 0.151 (0.01) & 0.091(0.004) \\
sVAR & 0.149 (0.055) & 0.123 (0.008) & 0.131 (0.009) & 0.086 (0.007) \\
BRRVAR  & 0.156 (0.004) &  0.153(0.007) & 0.158 (0.008) & 0.149 (0.004) \\
LASSOVAR& 0.111(0.008) &  0.099(0.012) &  0.088(0.018) & 0.068 (0.008) \\
HVARC & \textbf{0.107} (0.008) &  \textbf{0.086} (0.012) & \textbf{0.074} (0.018) &  \textbf{0.053} (0.008)\\
HVAROO & 0.115 (0.008) & 0.097(0.011) & 0.083 (0.017) & 0.069(0.009)\\
\hline \\[-1.8ex]
\footnotesize {\textit{Note:}}& \multicolumn{4}{l}{\footnotesize {$^*$ Parentheses contain estimated Standard Deviations.}}
\end{tabular}
\end{table}


\section{Conclusion} \label{s:s6}
 We have proposed the msVAR method, a modification of the two-stage sVAR method in  \citet{Davis2016}, where we substitute the first stage of the sVAR with the new and powerful time series graphical lasso algorithm to identify/estimate zeros of the inverse spectral density matrix while recognizing that its entries are complex-valued. The second stage implements refinement of the non-zero entries using FDR. This paper focuses on algorithmic and numerical results. Real data analysis and simulation results show usefulness of our method. 
Theoretical properties of our method, such as the consistency and the support recovery of the msVAR, are left for future research.

\clearpage

\appendix
\section{Appendix}

\subsection{Link Between AR coefficient and PSC} \label{s:arpsc}

As discussed, the assertion (\ref{eq:assert}) is not exact and can be violated for some AR models. Here, relying on the framework developed in \citet{songsiri2009}, we discuss the conditions when (\ref{eq:assert}) is exact.

The spectral density of the VAR process can be expressed as
\[f^Y(\omega) = \mathbf{A}^{-1}(e^{i\omega})\Sigma \mathbf{A}^{-H}(e^{i\omega}), \]
where $\mathbf{A}^{-1}(z)$ is the transfer function from $\omega$ to $Y$ , $\mathbf{A} = I + z^{-1}A_1 + \dots + z^{-p}A_p$, $B^H$ is the Hermitian transpose of matrix $B$ and $i = \sqrt{-1}$. Therefore, the inverse spectrum of an AR process is a trigonometric matrix polynomial
\begin{equation} \label{eq:cnd}
\Theta^Y(\omega) =  \mathbf{A}^{H}(e^{i\omega})\Theta \mathbf{A}(e^{i\omega}) = X_0 + \sum_{k = 1}^p(e^{-ik\omega}X_k + e^{ik\omega}X^{T}_k),
\end{equation}
where $\Theta = \Sigma^{-1}$, $X_k = \sum_{i = 0}^{p - k}A^T_i \Theta A_{i + k}$ with $A_0 = I$. From (\ref{eq:spd}) and (\ref{eq:cnd}) we obtain,
\begin{equation}
\mbox{PSC}_{ij}(\omega) = 0 \Leftrightarrow  (X_k)_{ij} = 0\; \mbox{for}\; k = 0, \dots, p.
\end{equation}

%

\section{Derivation of updates using Wirtinger Calculus} \label{a:wc}
Before providing details on solving updates (a) and (b) for (\ref{eq:14}), we give a brief overview of Wirtinger calculus. Deeper treatment of the subject can be found in \citet{remmert1991, kreutz2009} and in a pithy presentation by \citet{brandwood1983}.
\subsection{Wirtinger Calculus}
Let $z = x + iy$, where $x,y$ are real and $i = \sqrt{-1}$. Consider a general complex-valued function $f(z) = u(x,y) + i v(x,y)$, where we assume that the partial derivatives of $u$ and $v$ exist. Then the standard complex derivative $f^'(z)$ exist if $f(z)$ is holomorphic or Cauchy-Riemann equations are satisfied, i.e.,
\[ \frac{\partial u}{\partial v} = \frac{\partial v}{\partial y}, \;\; \frac{\partial v}{\partial x} = - \frac{\partial u}{\partial y} \]

Unfortunately, those conditions are strong, and the functions that we are usually interested in violate them. For example, $f(z) = |z|^2 = z^*z$, where $z^*$ is the conjugate of $z$. However, since real partial derivatives of a non-holomorphic function exist, one can exploit the real $R^2$ vector space structure, which underlies $C$, and represent $f(z) = f(x,y): R^2 \rightarrow R$. \citet{remmert1991} called the differentiation of this function $R$-derivative to avoid confusion with the standard complex derivative. As discussed in \citet[Section 3.1]{kreutz2009}, this representation can not be viewed as an admissible generalization of the standard complex derivative, since it, as well, suffers from some drawbacks. For example, it does not reduce to the standard complex derivative when a function $f(z)$ is holomorphic.

The generalization, in a sense discussed above, were developed in the notion of Wirtinger calculus \citep{wirtinger1927, brandwood1983}. In particular, a complex function $f(z)$ is viewed as a function of $z$ and its conjugate $z^*$
\[f(z) = f(z, z^*) = u(x,y) + i v(x,y).\]
It can be shown that $f(z,z^*)$ is holomorphic in $z$ for fixed $z^*$, and, similarly, holomorphic in $z^*$ for fixed $z$. Then the Wirtinger derivative and its conjugate are defined as
\[\frac{\partial f(z,z^*)}{\partial z}, \;\; \frac{\partial f(z,z^*)}{\partial z^*}\]
For example, for the function $f(z) = |z|^2 = z^*z$, $\frac{\partial f(z,z^*)}{\partial z} = z^*$ and $\frac{\partial f(z,z^*)}{\partial z*} = z$.

\subsection{Solving Updates (a) and (b)} \label{s:solveAB}
The derived formulas for updates (a) and (b) are given in (\ref{eq:updA}) and (\ref{eq:updB}), respectively.
It is instructive to note that, $L_\rho(\Theta[\cdot],Z[\cdot],U[\cdot])$ is separable in $n$ and the update of (a) can be implemented in parallel by minimizing
$J(\Theta[n]),\, n = 1, \dots, N$, where
\begin{equation} \label{eq:jtheta}
J(\Theta[n]) =  - \log \mbox{det}(\Theta[n]) + \mbox{tr} (\tilde f[n]\Theta[n]) + \rho/2 \|\Theta[n] - Z^{(k)}[n] + U^{(k)}[n]\|^2_F
\end{equation}
To use Wirtinger calculus, we write $J(\Theta[n]) = J(\Theta[n], \Theta^*[n])$ as a function of $\Theta$ and its complex conjugate $\Theta^*$. Thus, (\ref{eq:jtheta}) can be written as
\begin{equation}
\begin{aligned}
J(\Theta[n], \Theta^*[n]) &=  - \frac{1}{2}\Big[\log \mbox{det}(\Theta[n]) + \log \mbox{det} (\Theta^*[n]) + \mbox{tr} (\tilde f[n]\Theta[n]) + \mbox{tr} (\tilde f^*[n]\Theta^*[n]) \\
& + \rho \mbox{tr}(\Theta[n] - Z^{(k)}[n] + U^{(k)}[n]) (\Theta[n] - Z^{(k)}[n] + U^{(k)}[n])^H )\Big]
\end{aligned}
\end{equation}
Then in update (a), a necessary and sufficient condition for a global optimum  is that the gradient $J(\Theta[n], \Theta^*[n])$ with respect to $\Theta^*[n])$ is zero \citep{brandwood1983}:
\begin{equation} \label{eq:agrad}
	\Theta^{-1}[n] + \tilde f[n] + \rho(\Theta[n] - Z^{(k)}[n] + U^{(k)}[n]) = 0
\end{equation}

The solution to (\ref{eq:agrad}) follows as in \citet[Section 6.5]{Boyd2011}. Let the eigen-decomposition of the matrix $\rho(Z^{(k)}[n] - U^{(k)}[n]) - \tilde f[n]$ be $V_n C_n V^H_n $. Then
\begin{equation} \label{eq:eigen}
\Theta^{(k + 1)} = V_n \tilde C_n V^H_n ,
\end{equation}
where $\tilde C_n$ is the diagonal matrix with the $j$th diagonal element
\begin{equation} \label{eq:updA}
(\tilde C_n)_{jj} = (1/2\rho) (-(C_n)_{jj} + \sqrt{(C_n)_{jj} +4\rho}).
\end{equation}
 Finally, the update (a) is completed by obtaining the preceding solution for $n = 1,2, \dots, N$.

For an update (b), we consider the following two lemmas. The first lemma derives the Wirtinger subgradient for the penalty term in (\ref{tsggl:1}), and the second lemma provides an update for (b).

\begin{lemma} \label{l:subg}
Given $x \in C^N$, the Wirtinger subgradient of the function $T:C^N \rightarrow R$, where $T(x) = \|x\|_2$ is
\begin{equation}
\partial^W T(x) = \begin{cases}
	\frac{x}{2 \|x\|_2}, & \mbox{if }\, x \neq 0 \\
	\in \{u |\, \|u\|_2 \leq 1/ 2,\, u \in C^p\}, & \mbox{if }\,  x = 0
	
\end{cases}
\end{equation}
\end{lemma}

\begin{proof} We start from the $x \neq 0$ case. The derivative of $T(x, x^*) =  \|x\|_2 = (x^H x)^{1/2} $ with respect to the conjugate $x^*$ is
\[\frac{\partial T(x)}{x^*} = \frac{x}{2\|x\|_2}\]

For the case $x = 0$, to find the subgradient $\partial^W T(x^*)$ of the function $T(x)$, we exploit \citet[Definition 3]{bouboulis2012}. From which, $\partial^W T(x)$ is a Wirtinger subgradient if it satisfies

\begin{equation} \label{e:defsub}
T(y) \geq T(x) + 2\mbox{Re}((y - x)^H (\partial^W T(x))^*), \; y \in C^N
\end{equation}
where $\mbox{Re}(\cdot)$ indicates the real part of the complex variable. From (\ref{e:defsub}), we have for $x = 0$ and any $ y \in C^N$

\begin{equation} \label{eq:subg}
 \|y\|_2 \geq 2\mbox{Re}((y)^H (\partial^W T(x))^* )
\end{equation}
But (\ref{eq:subg}) is just the definition of the dual function of $\|x\|_2$, which is also $\ell_2$ norm \citep{Horn2012}, and the result follows.
\end{proof}

For the next lemma, we define the generic function $h: C^N \rightarrow R$
\begin{equation}
	h(x) = \frac{1}{2}\|a - x\|^2_2 + \lambda \|x\|_2
\end{equation}
and denote by $\hat x = \argmin_{x}h(x)$.

\begin{lemma} \label{l:comp}
	The $i$th component of the global minimum of $h(x)$ has the following closed form solution
	\begin{equation}
		\hat x_i = S_{\lambda/ \rho} (a_i),
	\end{equation}
	where $S_{\mu}(a_i) = (1 - \mu / \|a\|_2)a_i$
\end{lemma}

\begin{proof}
The proof of the lemma  relies on the framework developed in \citet{Friedman2010}, Lemma~\ref{l:subg} and Wirtinger calculus. Since $h(x)$ is convex on $x$, a necessary and sufficient condition for a global minimum $x^*$ is that the Wirtinger subdifferntial $\partial^W h(x) \in 0$. Thus, we solve
\[ 0 \in \frac{1}{2}(x - a) + \gamma,\]
where $\gamma$ is a subgradient from Lemma~\ref{l:subg}. Then, the result can be derived following steps as in \citet[Section 2]{Friedman2010}.

\end{proof}
 Invoking the Lemma~\ref{l:comp}, the update (b) for $i \neq j$ is

\begin{equation} \label{eq:updB}
	Z^{(k + 1)}_{ij}[n] = S_{\lambda/ \rho}(\Theta^{(k+1)}_{ij}[n] + U_{ij}^{(k)}[n])
\end{equation}
and $Z^{(k + 1)}_{ij}[n] = \Theta^{(k+1)}_{ii}[n]$ for $i = j$, since we do not penalize diagonal elements.

 \clearpage
\bibliography{ml_bib}
\bibliographystyle{te}

\end{document}